\documentclass[aps,superscriptaddress,twocolumn,twoside,floatfix,pra,nofootinbib,a4paper]{revtex4-2}
\usepackage{enumerate,appendix}
\usepackage{amsmath, amsthm, commath}
\usepackage{color,calc,graphicx}
\usepackage[dvipsnames,svgnames,table,cmyk]{xcolor}
\usepackage[colorlinks]{hyperref}
\hypersetup{
	colorlinks = true,
	urlcolor = {blue},
	citecolor = {magenta},
	linkcolor= {blue}
}

\usepackage{graphicx}
\usepackage{latexsym}
\usepackage{bbm}
\usepackage{tikz}
\usepackage{comment}
\usepackage{physics}
\usepackage[charter,cal=cmcal,sfscaled=false]{mathdesign}
\usepackage{booktabs}
\usepackage{multirow}
\usepackage{subfigure}
\usepackage{dcolumn}
\usepackage{diagbox}
\usepackage{array}
\usepackage{makecell}
\usepackage{threeparttable}
\usepackage[normalem]{ulem}
\usepackage{mdframed}
\usepackage{thmtools} 
\usepackage{thm-restate}
\usepackage[capitalize,nameinlink]{cleveref}
\usepackage{etoolbox}

\newtheorem{Theorem}{Theorem}
\newtheorem{corollary}{Corollary}
\newtheorem{lemma}{Lemma}

\newtheorem{definition}{Definition}
\newtheorem{result}{Result}

\theoremstyle{remark}
\newtheorem{remark}{Remark}

\begin{document}
\title{Non-signaling assistance in prepare-and-measure scenarios with classical communication}

\author{José Nogueira}
\email{josenogueira.castro@gmail.com}
\affiliation{Instituto de Física Gleb Wataghin, Universidade Estadual de Campinas, CEP 13083-859 Campinas, Brazil}

\author{Carlos Vieira}
\affiliation{\mbox{Instituto de Matemática, Estatística e Computação Científica, Universidade Estadual de Campinas, CEP 13083-859, Campinas, Brazil}}
\affiliation{Sorbonne Université, CNRS, LIP6, F-75005 Paris, France}    

\author{Lucas E. A. Porto}
\affiliation{Sorbonne Université, CNRS, LIP6, F-75005 Paris, France}

\author{Lucas Pollyceno}
\affiliation{\mbox{International Centre for Theory of Quantum Technologies, University of Gda{\'n}sk, 80-308 Gda\'nsk, Poland}}

\author{Rafael Rabelo}
\affiliation{Instituto de Física Gleb Wataghin, Universidade Estadual de Campinas, CEP 13083-859 Campinas, Brazil}

\author{Otfried Gühne}
\affiliation{Naturwissenschaftlich-Technische Fakult\"at, Universit\"at Siegen, Walter-Flex-Stra{\ss}e 3, 57068 Siegen, Germany}

\begin{abstract}
Extracting the full power of non-local correlations in prepare-and-measure (PM) scenarios requires precise control over the timing and structure of the receiver's measurements. 
Indeed, recent developments in entanglement-assisted classical communication scenarios have shown that adaptive strategies—where the receiver uses the transmitted message to guide their measurement choice—can outperform standard non-adaptive protocols. 
Moving beyond quantum theory, however, the ultimate limits of such advantages remain largely unexplored.
In this work, we thoroughly study adaptive and non-adaptive non-signaling (NS) assistance in PM scenarios with classical communication. 
We provide simple characterizations of the sets of behaviors that can be realized using both non-adaptive and adaptive NS assistance in arbitrary PM scenarios.
As a consequence, we show that non-adaptive NS assistance is already strong enough to reproduce quantum communication with the same message dimension: the transmission of a qudit can be simulated by a classical dit assisted non-adaptively by NS correlations.
We then compare adaptive and non-adaptive NS assistance. We prove that any adaptive NS advantage can be traced back to scenarios in which the receiver has no measurement choice, ruling out the genuinely multi-setting advantages found in entanglement-assisted quantum protocols.
Finally, we identify all PM scenarios where adaptive NS strategies provide a strict advantage over non-adaptive ones.
\end{abstract}

\maketitle

\section{Introduction}
The principle of non-signaling (NS) states that the choice of measurement performed by one observer cannot affect the statistics observed by a distant party~\cite{popescu_quantum_1994}. Although all quantum behaviors in Bell scenarios naturally satisfy this fundamental constraint, there exist valid NS correlations that cannot be realized within quantum theory~\cite{Brunner:2014RMP}. In spite of this fact, the formal study of the set of NS correlations has been proven extremely fruitful, as it provides a framework to investigate the set of quantum correlations from the outside and in a more theory-independent manner.

On the one hand, from a foundational perspective, a better understanding of NS correlations helps to clarify which properties are genuinely quantum and which arise purely from the NS constraints~\cite{masanes_general_2006, barnum_generalized_2007, barnum_teleportation_2008, skrzypczyk_emergence_2009}.
This approach has also motivated a search for additional principles capable of recovering the quantum set within the broader NS set, a research line that has attracted significant attention over the past two decades~\cite{brassard_limit_2006, pawlowski_information_2009, navascues_glance_2010, fritz_local_2013, cabello_simple_2013, nogueira_unexpected_2025}.
On the other hand, from a more applied point of view, there exist device-independent protocols for key distribution and randomness generation that rely solely on the assumption that the existing correlations satisfy the NS constraints \cite{acin_bells_2006, pironio_random_2010}.
Taken together, these developments, both foundational and applied, stem directly from the introduction of NS correlations and the subsequent study of their structural and operational properties.

In this work, we investigate the advantages enabled by the assistance of NS correlations in prepare-and-measure (PM) scenarios \cite{2026arXiv260323604B}.
At their core, these scenarios consist of a preparation device that encodes information into a physical system and transmits it to a measurement device, which then extracts data from the received system. Owing to their structure, PM scenarios can be regarded as fundamental building blocks of communication tasks, and for this reason they find a wide range of applications \cite{gallego2010pam, ahrens_experimental_2012, pawlowski2011semi, woodhead_secrecy_2015, brunner_dimension_2013, li_semi-device-independent_2011, li_semi-device-independent_2012, tavakoli_self-testing_2020, chaves2015information, egelhaaf_certifying_2025, Bowles_Certifying_2014, monogamy_IC, raj2025waveparticle}, which include protocols such as dense coding \cite{bennet1992communication} and random access codes \cite{wiesner1983conjugate,ambainis_dense_1999}.

Although different forms of communication can be considered in PM scenarios, in this work we consider specifically scenarios where the communication is classical.
In this setting, the assistance of a non-signaling box shared between sender and receiver can be conceived under two different classes of protocols, depending on how the classical message is used by the receiver.
More precisely, we can distinguish between an \textit{adaptive} use of the message—where the receiver's input to the NS box depends on the incoming classical message—and a \textit{non-adaptive} use of the message, where the receiver's input is chosen independently of the message, which is only incorporated during a later classical post-processing step. 
Interestingly, it has been recently shown that in entanglement-assisted protocols with classical communication, adaptive strategies are strictly more powerful than non-adaptive ones~\cite{Pauwels_adaptive_2022}.

Despite the growing interest in PM scenarios with different forms of communication and shared resources, the role of arbitrary NS correlations as assistance remains comparatively unexplored. Beyond their non-adaptive use in certain communication scenarios \cite{pawlowski_information_2009, van_dam_implausible_2013}, only a few recent works have begun to explore adaptive NS-assisted protocols \cite{Van_Himbeeck_2019, rout_adaptivePM_2025, pollyceno2026}. 
The goal of the present work is precisely to develop a more systematic study of NS assistance in PM scenarios, and to explore and understand when adaptive strategies outperform non-adaptive ones.

First, we discuss how to characterize the sets of behaviors achievable with non-adaptive and adaptive NS assistance in arbitrary PM scenarios. 
We show that both characterizations can be reduced to smaller PM scenarios, in which the receiver has no inputs. 
In the non-adaptive case, this implies that the constraints defining the achievable set are exactly those defining the classical set of behaviors in this reduced scenario; see \cref{res:NS_equals_C}. 
As a consequence, by combining this reduction with the result of Frenkel and Weiner~\cite{frenkel-storage-2015}, we show that non-adaptive NS assistance can reproduce quantum communication with the same message dimension: the transmission of an $n_A$-dimensional quantum system can be simulated by an $n_A$-valued classical message assisted by arbitrary NS correlations; see \cref{cor:qudit_simulation_NS}. 
Meanwhile, in the adaptive case, the achievable behaviors are characterized by a simple family of inequalities with an information-theoretic interpretation~\cite{Chitambar2023CommunicationValueOfQuantumChannel}; see \cref{res:NS_equals_G}. 
To prove this adaptive characterization, we observe that the problem of deciding whether a behavior can be realized in a PM scenario where the receiver has no inputs can be rephrased as the problem of simulating classical channels with zero-error communication, which has been extensively studied in the literature~\cite{Bennett2001ReverseShannon, Fang2018ChannelSimulation, Wang2016QuantumNSAssisted, Duan2016NSAssisted, Berta2013EntanglementCost,Berta2011QuantumReverseShannonWithOne-Shot,Doolittle_2021_Certifying,Cubitt:2011TIT,Cubitt:2010PRL,Chitambar2023CommunicationValueOfQuantumChannel}. 
In particular, the results of Refs.~\cite{Doolittle_2021_Certifying,Cubitt:2011TIT} are of great value for this work.

We then use these characterizations to compare adaptive and non-adaptive NS assistance. 
This allows us to identify all PM scenarios in which adaptivity provides a strict advantage over non-adaptive protocols; see \cref{res:classification}. 
Moreover, we show that any such advantage can already be witnessed in the reduced scenarios without receiver inputs; see \cref{cor:main}. 
Thus, in contrast with the quantum case, arbitrary NS assistance does not lead to genuinely multi-setting adaptive advantages.

\section{Preliminaries}\label{sec:preliminaries}

\subsection{Prepare-and-measure scenarios}
A prepare-and-measure scenario can be viewed as a correlation scenario capturing the operational structure of general one-way communication tasks \cite{gallego2010pam}. In its standard form, two parties are involved: Alice and Bob.
Before the protocol starts, Alice and Bob may agree on a strategy, or establish more general correlations between them to implement the different steps of the protocol. 
The protocol then starts when Alice receives an input labeled by $x \in [n_{X}] := \{0, 1,\ldots, n_{X} - 1\}$, and Bob an input labeled by $y \in [n_{Y}] := \{0, 1, \ldots, n_{Y} - 1\}$.
Based on $x \in [n_{X}]$ and on the potential pre-established correlations, Alice prepares a message which is sent to Bob. 
Finally, based on his input $y$, the received message and the potential pre-established correlations, Bob produces an output $b \in [n_{B}] := \{0, 1, \ldots, n_{B} - 1\}$. 
The PM scenario is thus specified by the number of Alice's and Bob's possible inputs, and the number of Bob's possible outputs, which we denote by $(n_X; n_Y, n_B)$.
See \cref{fig:PM_scenario} for an illustration.

\begin{figure}[t]
    \centering
     \includegraphics[width=\linewidth]{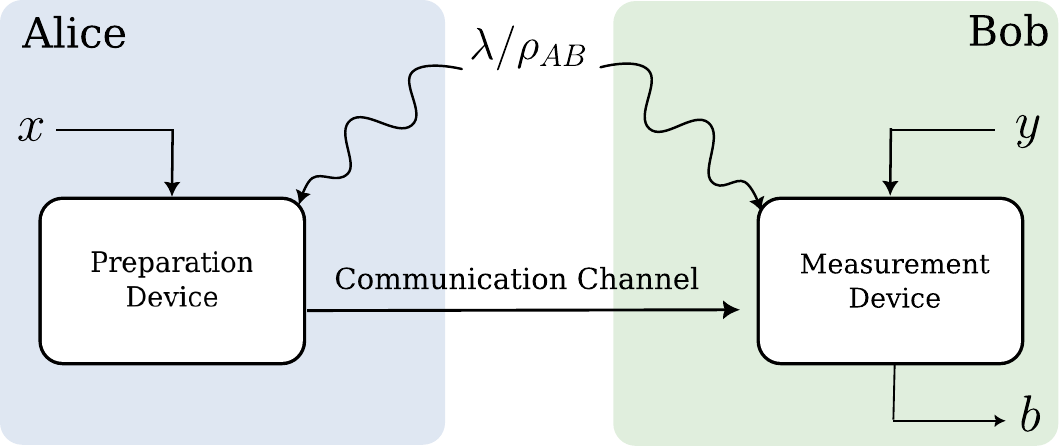}
    \caption{
        Schematic of a prepare-and-measure scenario. Alice encodes her input $x$ into a preparation device, which emits either a classical or quantum system sent to Bob. Bob receives the system, chooses an input $y$ for his measurement device, and produces an output $b$. The devices may share correlations, which can be classical (shared randomness $\lambda$) or quantum (arising from a shared entangled state $\rho_{AB}$).
    }\label{fig:PM_scenario}
\end{figure}

The procedure described above corresponds to a single round of the experiment. After many repetitions, Alice and Bob can estimate the conditional probabilities $\mathbf{q} = [q(b|xy)]_{b, x, y}$ that characterize the observed statistics. We call $\mathbf{q}$ a \textit{behavior} of the PM scenario. It is often useful to geometrically think of a behavior as a vector in $\mathbb{R}^{n_B n_X n_Y}$, with entries $q(b\vert xy)$.
In this sense, the set of behaviors that Alice and Bob can generate with unlimited resources correspond to all $\mathbf{q} = [q(b|xy)]_{b, x, y}$ satisfying $q(b \vert x y) \geq 0$ and $\sum_{b = 0}^{n_B - 1} q(b \vert xy) = 1 ~\forall x \in [n_X], y \in [n_y]$.

A fundamental question in PM scenarios is to characterize the set of behaviors that Alice and Bob can generate using certain kinds of resources. This set of achievable correlations depends on: (i) the physical system used as the message (including its dimensionality); (ii) the nature of the pre-established correlations between the parties; and (iii) the way in which Bob uses the received message to produce his output $b$.

Regarding the first aspect, throughout this work we restrict our analysis to classical communication with a bounded alphabet size. More precisely, the message prepared by Alice is always a classical symbol $a \in [n_A] = \{0, \ldots, n_A-1\}$. Notice that if the dimension of the message $n_A$ is greater than or equal to $n_X$, Alice can always communicate her input to Bob $a = x$, and any correlation $q(b \vert xy)$ can be constructed by the parties.
For this reason, the problem of understanding the correlations that Alice and Bob can generate with classical communication is nontrivial only when $n_A < n_X$~\cite{gallego2010pam}.

As for the second point, the simplest form of pre-shared correlations that may arise are also due to classical systems, what is known as shared randomness. In this case, a behavior $\mathbf{q} = [q(b | xy)]$ can be realized with an $n_A$-dimensional classical message if there exist probability distributions $\{\pi(\lambda)\}_\lambda$, $\{p(a \vert x \lambda)\}_{a, x, \lambda}$ and $\{p(b \vert ya \lambda)\}_{b, y, a, \lambda}$ such that
\begin{equation}\label{eq:PM_classical}
q(b|xy) = \sum_{a=0}^{n_{A}-1} \sum_{\lambda} \pi(\lambda) p(a|x \lambda) p(b|ya \lambda).
\end{equation}
The behaviors that admit such a decomposition are referred to as \textit{classical behaviors}. The set of all classical behaviors that can be realized with the communication of an $n_A$-dimensional classical message in the scenario $(n_X; n_Y, n_B)$ is denoted by $\mathrm{C}(n_{X}; n_{Y}, n_{B})_ {n_{A}}$.
This set corresponds to a polytope in $\mathbb{R}^{n_B n_X n_Y}$, and the problem of deciding whether a given behavior is classical can be cast as a linear program (LP) \cite{gallego2010pam}.

A more sophisticated possibility for the type of correlations shared between Alice and Bob arises when a bipartite entangled state $\rho_{AB}$ is distributed to them. In this case, the most general way in which they can generate behaviors within a PM scenario is as follows:
(a) Alice performs a measurement with $n_{A}$ possible outcomes on her subsystem of $\rho_{AB}$;
(b) she then sends the corresponding classical outcome to Bob; and
(c) upon receiving this message, and conditioned on his input $y$, Bob performs a measurement on his subsystem of $\rho_{AB}$ \cite{Entanglement_Marcin2010, tavakoli2021correlations, vieira_interplays_2023}. 
The resulting class of behaviors can therefore be expressed as:
\begin{equation}\label{eq:Q_PM_A}
q(b|xy) = \sum_{a=0}^{n_{A}-1} \Tr{\rho_{AB} (E_{a|x} \otimes F_{b|y,a})},
\end{equation}
where $E_{a|x}$ and $F_{b|y,a}$ are POVM elements~\cite{nielsen_chuang_2010}.

It is worth noting that Bob’s measurement operator, $F_{b|y,a}$, explicitly depends on his input $y$ as well as on the message $a$ received from Alice. This dependency indicates that Bob must postpone the choice of his measurement until Alice’s message arrives. Hence, his measurement strategy can be regarded as \textit{adaptive}, in the sense that it is conditioned on the information communicated by Alice \cite{adaptive2022pauwels}. Correlations of this form are referred to throughout this work as \textit{entanglement-assisted (EA) adaptive behaviors}.

In contrast, the \textit{non-adaptive} case refers to the situation where Bob’s measurement choice does not depend on Alice’s message. Instead, he chooses his measurement based solely on $y$, performs it on his subsystem to obtain an intermediate outcome $\beta$, and only at a later stage combines $\beta$ and the message $a$ through classical post-processing to produce his final output $b$. The class of behaviors obtained under this strategy is therefore given by:
\begin{equation}\label{eq:Q_PM_NA}
q(b|xy) = \sum_{a=0}^{n_{A}-1} \sum_{\beta} \Tr{\rho_{AB} (E_{a|x} \otimes F_{\beta|y})} p(b|a\beta).
\end{equation}
We refer to the distributions generated by this restricted strategy as \textit{EA non-adaptive behaviors}.

It is straightforward to see that non-adaptive strategies form a subset of adaptive ones. Indeed, \cref{eq:Q_PM_NA} can be recovered from \cref{eq:Q_PM_A} by defining Bob’s measurement operators as
\begin{equation}
    F_{b|y,a} = \sum_{\beta} p(b|a\beta)\, F_{\beta|y}.
\end{equation}
On the other hand, determining whether adaptive strategies are \textit{strictly} more powerful than non-adaptive ones is a less trivial question. While adaptivity yields no advantage in some standard protocols, such as the $2 \to 1$ random access code (RAC)~\cite{buhrman_quantum_2001}, there are specific instances, as discussed by Pauwels et al. \cite{adaptive2022pauwels}, where adaptive strategies strictly outperform non-adaptive ones. Similarly, 
advantages of adaptive strategies over non-adaptive ones have been found in the related tasks of quantum state verification \cite{yu_optimal_2019} and
the implementation of measurements via local operations and classical communication (LOCC) \cite{dutra_structure_2025}.

Finally, Alice’s and Bob’s boxes in the PM scenario may be correlated, or assisted, by general Bell-nonlocal correlations; that is, by resources such as Popescu–Rohrlich (PR) boxes \cite{popescu_quantum_1994} or, more generally, by any Bell-type box satisfying the non-signaling constraints \cite{Brunner:2014RMP}. In this work, our main focus will be precisely on this most general form of assistance, and on the distinctions between adaptive and non-adaptive strategies in this case.

\subsection{Non-adaptive non-signaling assistance}
In previous works on prepare-and-measure scenarios, non-signaling correlations were often used as assisting resources \cite{pawlowski_information_2009, van_dam_implausible_2013, tavakoli_does_2020}. 
Although not phrased in these terms, this assistance was non-adaptive: Bob's use of the shared resource was fixed independently of the message received from Alice.
In its most general form, assuming that Alice can communicate an $n_A$-dimensional classical message to Bob, a behavior generated under non-adaptive NS assistance can be expressed as
\begin{equation}\label{eq:old_non_adaptive}
\begin{split}
        q(b|xy) = \sum_{a = 0}^{n_A - 1}\sum_{\alpha, \beta, \mathcal{X}, \mathcal{Y}} 
        & p(\mathcal{X}|x) p(\mathcal{Y}|y) p(\alpha \beta|\mathcal{XY})  \\
        &\times  p(a|x \mathcal{X}\alpha) p(b|y a \mathcal{Y}\beta),
\end{split}
\end{equation}
where $p(\alpha \beta|\mathcal{X}\mathcal{Y})$ is a NS box shared between the parties; $p(\mathcal{X}|x)$ and $p(\mathcal{Y}|y)$ are local pre-processings by Alice and Bob, respectively, used to choose the inputs of the NS box; $p(a|x\mathcal{X}\alpha)$ represents Alice's post-processing to decide, from her inputs and the output of the assisting NS box, the classical message $a \in [n_A]$ sent to Bob; and $p(b|ya\mathcal{Y}\beta)$ corresponds to Bob's post-processing to decide his final outcome $b$.
Therefore, $\mathbf{q}$ admits a non-adaptive NS model if there exists a NS box and local pre- and post-processings recovering $\mathbf{q}$ through \cref{eq:old_non_adaptive}.

It is interesting to observe that, without loss of generality, the local pre-processings and Alice's post-processing can be absorbed into the assisting NS box.
Moreover, by increasing the number of Bob's outputs of the NS box, one can also make the post-processing of Bob independent of his input $y$.
Additionally, even if additional classical correlations (shared randomness) are given to the parties, they can also be absorbed into the NS box.
In this sense, there are many equivalent notions of non-adaptive NS models, which can be transformed into one another via the absorption of terms into the NS box. For simplicity, in the definition below we absorbing all we can into the NS box.
\begin{definition}\label{def:PM_behaviour_NS_NA} 
A PM behavior $\mathbf{q} = [q(b|xy)]$ can be realized with a classical message of alphabet size $n_A$ and non-adaptive NS assistance if it can be written as
\begin{equation}\label{eq:NA_NS_model_1} 
q(b|xy) = \sum_{a=0}^{n_A-1}\sum_{\beta} p(a\beta|xy)\, p(b|a\beta), 
\end{equation} 
where $p(a\beta|xy)$ is a bipartite NS box, the output $a \in [n_A]$ corresponds to the classical message sent by Alice to Bob, and $p(b|a\beta)$ is Bob's final output post-processing.

We denote by $\mathrm{NS}^{\mathrm{NA}}(n_{X}; n_{Y}, n_{B})_{n_{A}}$ the set of all behaviors of the PM scenario $(n_{X}; n_{Y}, n_{B})$ that admit such a decomposition. 
\end{definition}

Operationally, the process admits the following description. Alice inputs $x$ into her half of the NS box $p(a\beta|xy)$, which produces an output $a$ defining the classical message to be sent to Bob. Bob, in turn, inputs $y$ into his half of the NS box $p(a\beta|xy)$, yielding an output $\beta$. Finally, Bob uses $a$, and $\beta$ to generate the final output $b$ according to $p(b|a\beta)$.
See \cref{fig:non-adaptive_scen_fig} for an illustration.

Notice that non-adaptive EA strategies (\cref{eq:Q_PM_NA}) are a particular form of non-adaptive NS-assisted strategies. Thus, the set of non-adaptive NS-assisted behaviors provides an outer approximation to the set of behaviors achievable with non-adaptive entanglement assistance. As we will see, this outer approximation admits a much simpler characterization than its quantum counterpart.

\begin{figure}[t]
    \centering
    \includegraphics[width=\linewidth]{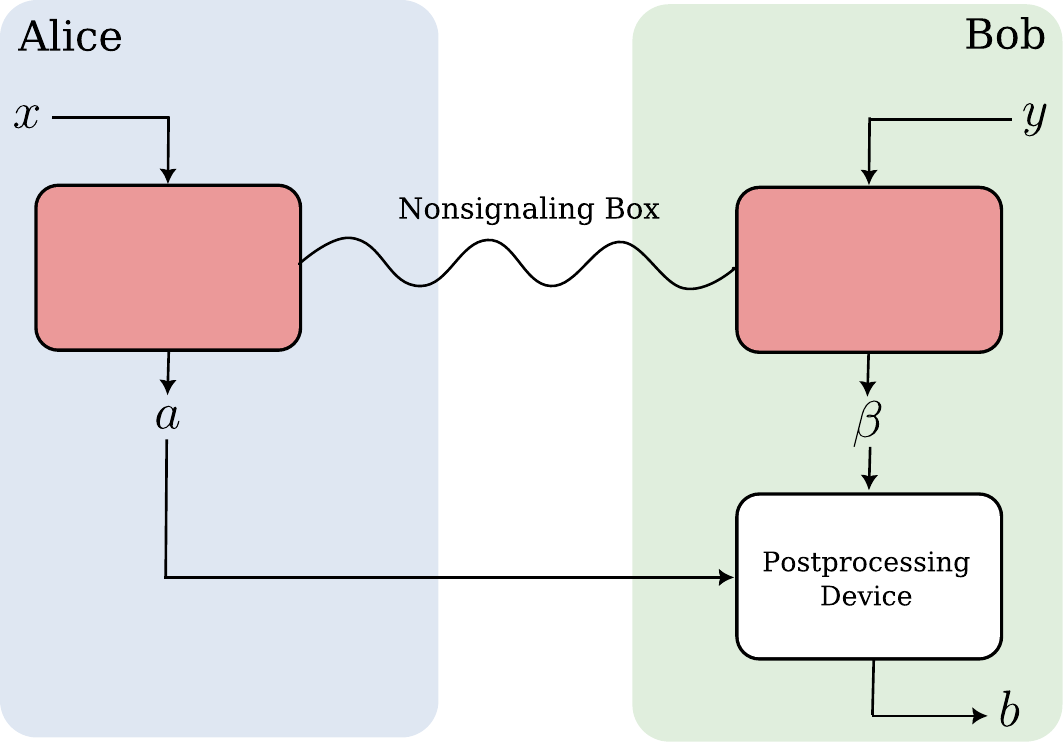}
    \caption{
        Schematic representation of a prepare-and-measure scenario with non-adaptive non-signaling assistance, according to the definition in \cref{eq:NA_NS_model_1}. Alice and Bob share an arbitrary non-signaling box in the Bell scenario $(n_X,n_A;n_Y,n_{\beta})$. Alice inputs $x$ into her share of the box, obtains outcome $a$, and communicates it to Bob. Bob independently inputs $y$ into his share of the box and obtains $\beta$. Finally, he post-processes $a,\beta$ to produce the output $b$. 
    }\label{fig:non-adaptive_scen_fig}
\end{figure}

\subsection{Adaptive non-signaling assistance}
The form of NS assistance discussed in the previous subsection is not, however, the most general one.
Rather than choosing his input to the shared box before receiving Alice's message, Bob may wait for the message to arrive and only then decide which input to use. The behaviors the parties can generate in this way when the dimension of Alice's message is $n_A$ have the form
\begin{equation}\label{eq:old_adaptive}
\begin{split}
        q(b|xy) = \sum_{a = 0}^{n_A - 1}\sum_{\alpha, \beta, \mathcal{X}, \mathcal{Y}} 
        & p(\mathcal{X}|x) p(\mathcal{Y}|ay) p(\alpha \beta|\mathcal{XY})  \\
        &\times  p(a|x \mathcal{X}\alpha) p(b|y a \mathcal{Y}\beta),
\end{split}
\end{equation}
where $p(\alpha \beta|\mathcal{X}\mathcal{Y})$ is a NS box shared between the parties; $p(\mathcal{X}|x)$ and $p(\mathcal{Y}|ay)$ represent local pre-processings; $p(a|x\mathcal{X}\alpha)$ and $p(b|ya\mathcal{Y}\beta)$ correspond to classical post-processings for Alice and Bob, respectively.
The difference with respect to the non-adaptive case \cref{eq:old_non_adaptive} is that now we allow for Bob's pre-processing to depend on the message $a$.

Similarly to the non-adaptive case, additional shared randomness and pre- and post-processing can be absorbed into the assisting NS box, thus also leading to many equivalent models of adaptive NS assistance.
However, unlike the non-adaptive case, in an adaptive model one can also absorb Bob's post-processing into the assisting NS box, thus leading to the following definition.

\begin{definition}\label{def:PM_behaviour_NS_A}
A PM behavior $\mathbf{q} = [q(b|xy)]$ can be realized with a classical message of alphabet size $n_A$ and adaptive NS assistance if it can be written as
\begin{equation}\label{eq:def_q_NS_PM}
   q(b|xy) = \sum_{a=0}^{n_A-1} p(ab|x,(y,z=a)),
\end{equation}
where $p(ab|x,(y,z))$ is a bipartite NS box, the output $a \in [n_A]$ corresponds to the classical message sent by Alice to Bob, and, after receiving $a$, Bob inputs the composite setting $(y,a)$ into his side of the box.

We denote by $\mathrm{NS}^{\mathrm{A}}(n_{X}; n_{Y}, n_{B})_{n_{A}}$ the set of all behaviors of the PM scenario $(n_{X}; n_{Y}, n_{B})$ that admit such a decomposition.
\end{definition}

Operationally, the process unfolds as follows. Alice and Bob share a NS box characterized by conditional probabilities $p(a,b|x,(y, z))$. Alice inputs $x$ into her part of the box and obtains the outcome $a$. This value is then sent to Bob as a classical message. Once Bob receives $a$, he provides the pair $(y, a)$ as the input to his side of the shared box, which then produces the outcome $b$. The corresponding PM behavior $q(b|xy)$ is obtained by averaging over all possible messages $a$, as described in \cref{eq:def_q_NS_PM}.
We note that in this form, the protocol corresponds precisely to the approach introduced by Rout \textit{et al.}~\cite{rout_adaptivePM_2025} to employ NS correlations as a resource in PM scenarios. An illustration is given in \cref{fig:adaptive_scen_fig}.

\begin{figure}[t]
    \centering
    \includegraphics[width=\linewidth]{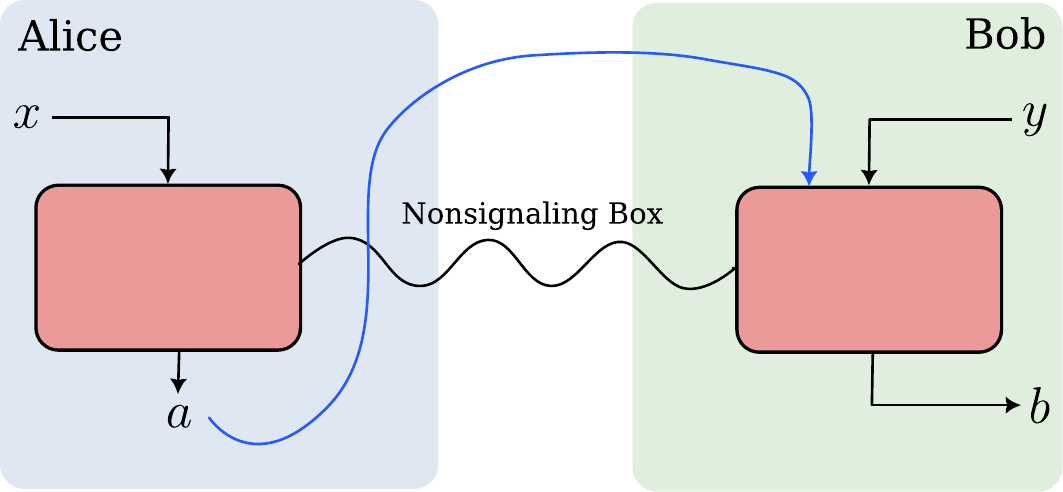}
    \caption{
    Schematic representation of a prepare-and-measure scenario with adaptive non-signaling assistance, according to \cref{eq:def_q_NS_PM}. Alice and Bob share an arbitrary non-signaling box in the Bell scenario $(n_X,n_A;n_Y n_A,n_B)$. Alice inputs $x$ into her share of the box, obtains outcome $a$, and communicates it to Bob. Bob then inputs the composite setting $(y,a)$ into his share of the box and obtains the output $b$. The protocol is adaptive because Bob's input to the shared box depends on the message received from Alice.
    }\label{fig:adaptive_scen_fig}
\end{figure}

This general framework naturally encompasses the previously discussed models. First, any PM behavior admitting an adaptive EA model (\cref{eq:Q_PM_A}) inherently admits an adaptive NS-assisted one, by defining the underlying bipartite box as $p(ab|x,(yz)) = \Tr{\rho_{AB} (E_{a|x} \otimes F_{b|y,z})}$. 
Second, within the NS regime itself, adaptive models are at least as powerful as non-adaptive ones, since Bob can always postpone using Alice's message until a final classical post-processing step. Therefore, $\mathrm{NS}^{\mathrm{NA}}(n_{X}; n_{Y}, n_{B})_{n_{A}} \subseteq \mathrm{NS}^{\mathrm{A}}(n_{X}; n_{Y}, n_{B})_{n_{A}}$, meaning that every PM behavior admitting a non-adaptive NS-assisted model inherently admits an adaptive one.

One of the central objectives of this work is to investigate whether and when this inclusion is strict; that is, whether and when an adaptive advantage of NS assistance over its non-adaptive use exists.
In the following, we discuss characterizations of NS assistance which allow us to completely identify all PM scenarios where adaptive strategies provide a strict advantage over non-adaptive ones.

\section{Results}\label{sec:results}
\subsection{A characterization of the non-adaptive NS-assisted set of behaviors for PM scenarios}\label{sub_sec:A}

In this section, we characterize the set of PM behaviors that can be realized with an $n_A$-dimensional classical message and non-adaptive assistance by NS correlations, $\mathrm{NS}^{\mathrm{NA}}(n_X; n_Y, n_B)_{n_A}$. We show that this characterization can be reduced to the classical set of behaviors $\mathrm{C}(n_X; 1, n_B)_{n_A}$ associated with the simpler PM scenario in which Bob has no inputs. Beyond its structural implications, this reduction reveals a direct connection with the problem of simulating classical channels with zero-error communication, a topic that has been extensively studied in the literature~\cite{Bennett2001ReverseShannon,Fang2018ChannelSimulation, Wang2016QuantumNSAssisted, Duan2016NSAssisted,Berta2013EntanglementCost,Berta2011QuantumReverseShannonWithOne-Shot, Doolittle_2021_Certifying, Cubitt:2011TIT, Cubitt:2010PRL, Chitambar2023CommunicationValueOfQuantumChannel}.

According to \cref{def:PM_behaviour_NS_NA}, in a non-adaptive protocol Bob must perform his measurement independently of Alice’s message. 
Operationally, this is equivalent to Bob producing an entire tuple of potential outputs in advance—one specific response for each possible message $a$ he might later receive. 
This can be captured by introducing a bipartite NS box $p(a; b_{0}\ldots b_{n_{A}-1}|x;y)$, where the output on Bob's side is a string $b_{0}\ldots b_{n_{A}-1}$ of length $n_A$.
Moreover, by twirling the outcomes of Alice, we can take this box to have uniform marginals on Alice without loss of generality.
These facts are formalized in the following lemma, which is proven in Appendix~\ref{appdx:proofPropo}.

\begin{restatable}[]{lemma}{LemmaNA}\label{lemma:NA}
    A PM behavior $\mathbf{q} = [q(b|xy)]$ admits a non-adaptive NS-assisted model if, and only if, there exists a bipartite NS box $p(a; b_{0}\ldots b_{n_{A}-1}|x;y)]$ such that
    \begin{equation}\label{eq:NA_NS_model_2}
        q(b|xy) = \sum_{a =0}^{n_A -1 }\sum_{b_0\ldots b_{n_{A}-1}} \delta_{b_a , b}\, p(a; b_{0}\ldots b_{n_{A}-1}|x;y).   
    \end{equation}
    Furthermore, this underlying NS box can always be chosen to have a uniform marginal for Alice, i.e.,
    \begin{equation}
        \sum_{b_0 ... b_{n_{A}-1}} p(a; b_0 ... b_{n_{A}-1}|x;y) = \frac{1}{n_{A}}.    
    \end{equation}
\end{restatable}

The above lemma provides a useful representation of non-adaptive NS assistance. In particular, unlike \cref{eq:NA_NS_model_1}, where the auxiliary output $\beta$ is not a priori bounded in cardinality, \cref{eq:NA_NS_model_2} yields a finite-dimensional description of the set $\mathrm{NS}^{\mathrm{NA}}(n_{X}; n_{Y}, n_{B})_{n_{A}}$. Consequently, membership in this set can be decided via a linear program.
Additionally, it is interesting to compare \cref{eq:NA_NS_model_2} with the corresponding \cref{eq:def_q_NS_PM} for the adaptive case.
While the latter requires assistance from a non-signaling box with $n_Y n_A$ inputs for Bob, the former requires a non-signaling box where the number of outputs is $n_B n_A$, highlighting the step of the protocol where Bob can use the classical message sent by Alice.

\cref{lemma:NA} also provides the essential ingredient to prove the following theorem, which shows that the problem of characterizing non-adaptive NS assistance in arbitrary scenarios can be reduced to PM scenarios where Bob has no inputs.
\begin{Theorem}~\label{thm:Porto_thm}
    Let $\mathbf{q} = [q(b|xy)]$ be a behavior of the PM scenario $(n_X; n_Y, n_B)$, and let $\mathbf{q}_y$ denote the behaviors of the scenario $(n_X; 1, n_B)$ constructed via $q_{y}(b|x) := q(b|xy)$. 
    Then, $\mathbf{q} \in \mathrm{NS}^{\mathrm{NA}}(n_{X}; n_{Y}, n_{B})_{n_{A}}$ if, and only if, $\mathbf{q}_y \in \mathrm{NS}^{\mathrm{NA}}(n_{X}; 1, n_{B})_{n_{A}}$ for all $y \in [n_y]$.

    In other words, the behavior $\mathbf{q}$ can be realized with the communication of an $n_A$-dimensional classical message and non-adaptive NS assistance if, and only if, each of its sub-behaviors $\mathbf{q}_y$ also can.
\begin{proof}
    If $\mathbf{q} \in \mathrm{NS}^{\mathrm{NA}}(n_{X}; n_{Y}, n_{B})_{n_{A}}$, it naturally follows that $\mathbf{q}_y \in \mathrm{NS}^{\mathrm{NA}}(n_{X}; 1, n_{B})_{n_{A}}$ for each of the sub-behaviors $\mathbf{q}_y$.
    Therefore, to prove the theorem we just need to prove the reverse direction, that is, we need to prove that if each sub-behavior admits a non-adaptive NS-assisted model, then the full behavior $\mathbf{q}$ also admits such a model.
    To do so, notice that if for every input $y$, the sub-behavior $\mathbf{q}_y$ admits a non-adaptive NS-assisted model with an $n_A$-dimensional classical message, by \cref{lemma:NA} for each $y$ there exists an underlying NS box $p_{y}(a; b_0 \dots b_{n_A-1}|x)$ that perfectly reproduces $\mathbf{q}_y$ via \cref{eq:NA_NS_model_2}. Crucially, \cref{lemma:NA} guarantees that each of these boxes can be taken to have perfectly uniform marginals for Alice, i.e.,
    \begin{equation}\label{eq:p_y_has_uniform_marginals_on_Alice}
        \sum_{b_0 \dots b_{n_A-1}} p_{y}(a; b_0 \dots b_{n_A-1}|x) = \frac{1}{n_A} \quad \forall y.
    \end{equation}
    We can then construct a full bipartite box for the entire scenario by defining its components as
    \begin{equation}
        p(a; b_0 \dots b_{n_A-1}|x;y) := p_{y}(a; b_0 \dots b_{n_A-1}|x).
    \end{equation}
    By construction, this joint box perfectly reproduces the full behavior $\mathbf{q}$ via \cref{eq:NA_NS_model_2}, and one can verify that it indeed satisfies the non-signaling constraints (non-signaling from Bob to Alice is guaranteed by \cref{eq:p_y_has_uniform_marginals_on_Alice}). Therefore, since $p(a; b_0 \dots b_{n_A-1}|x;y)$ is a valid non-signaling box that successfully recovers $q(b|xy)$, we conclude that the full behavior $\mathbf{q}$ admits a non-adaptive NS-assisted model.
    \end{proof}
\end{Theorem}

\cref{thm:Porto_thm} reduces the problem of characterizing non-adaptive NS assistance in arbitrary PM scenarios to the study of scenarios in which Bob has no inputs. We therefore turn our attention to the set $\mathrm{NS}^{\mathrm{NA}}(n_{X};1,n_{B})_{n_{A}}$.

At first sight, one may expect non-adaptive NS-assistance to provide no advantage in such scenarios. Indeed, when one of the parties has only a single input, nonlocality cannot arise: every NS box admits a local-hidden-variable model~\cite{Brunner:2014RMP}. As we now show, this intuition is correct, and non-adaptive NS-assisted models coincide exactly with classical models.

According to \cref{def:PM_behaviour_NS_NA}, a behavior $\mathbf{q}=[q(b|x)]$ that admits a non-adaptive NS-assisted model can be written as
\begin{equation}\label{eq:PM_behaviour_Channel2}
\begin{split}
    q(b|x) = \sum_{a = 0}^{n_A - 1}\sum_{\beta} 
    	p(a\beta|x)\, 
        p(b|a\beta).
\end{split}
\end{equation}
Using Bayes’ rule together with the NS constraints, we can rewrite the shared NS box as $p(a\beta|x) = p(\beta)p(a \vert x \beta)$. Hence,
\begin{equation}
    q(b|x)
    =
    \sum_{a=0}^{n_A-1}\sum_{\beta}
    p(\beta)\,
    p(a|x\beta)\,
    p(b|a\beta),
\end{equation}
which is precisely a classical model in the PM scenario (see \cref{eq:PM_classical}), where $\beta$ plays the role of the hidden variable. We have thus proven the following remark.
\begin{remark}\label{remark_NSeqClassical_Yeq1}
In PM scenarios where Bob has no inputs (i.e., where $n_{y}=1$) every behavior that admits a non-adaptive NS-assisted model with $n_A$-dimensional classical communication admits a classical model (\cref{eq:PM_classical}) with the same communication, i.e.,
\begin{equation}
\mathrm{NS}^{\mathrm{NA}}(n_{X};1,n_{B})_{n_{A}} = \mathrm{C}(n_{X};1,n_{B})_{n_{A}}.
\end{equation}
\end{remark}

Combining this fact with \cref{thm:Porto_thm}, we can state the following result:
\begin{result}\label{res:NS_equals_C}
Let $\mathbf{q} = [q(b|xy)]$ be a behavior of the PM scenario $(n_X; n_Y, n_B)$, and let $\mathbf{q}_y$ be the behaviors of the scenario $(n_X; 1, n_B)$ constructed via $q_{y}(b|x) := q(b|xy)$. Then, $\mathbf{q} \in \mathrm{NS}^{\mathrm{NA}}(n_X;n_Y,n_B)_{n_A}$ if and only if $\mathbf{q}_y \in \mathrm{C}(n_X;1,n_B)_{n_A}$ for all $y \in [n_y]$. 
That is,
\begin{equation}\label{eq:non-adaptive_NS_set_alt_def}
\mathrm{NS}^{\mathrm{NA}}(n_X;n_Y,n_B)_{n_A}
= \big\{\mathbf{q} \:\big|\:
\mathbf{q}_{y} \in
\mathrm{C}(n_X;1,n_B)_{n_A}
\:\forall\:\ y\big\}.
\end{equation}

In other words, the behavior $\mathbf{q}$ can be realized with the communication of an $n_A$-dimensional classical message with non-adaptive NS-assistance if, and only if, each of its sub-behaviors $\mathbf{q}_y$ admits a classical model with an $n_A$-dimensional message.
\end{result}
Therefore, the characterization of non-adaptive NS assistance in arbitrary PM scenarios reduces to understanding the classical sets $\mathrm{C}(n_X;1,n_B)_{n_A}$, whose structure depends on the relation between $n_X$, $n_B$, and the message size $n_A$.

We now discuss some consequences of \cref{res:NS_equals_C}. 
The first one concerns quantum communication. 
Since \cref{res:NS_equals_C} reduces non-adaptive NS assistance to fixed-$y$ sub-behaviors, it can be combined with the result of Frenkel and Weiner \cite{frenkel-storage-2015} for PM scenarios with a fixed measurement to show that non-adaptive NS assistance reproduces quantum communication with the same message dimension. Indeed, let $\mathrm{Q}(n_X;n_Y,n_B)_{n_A}$ denote the set of PM behaviors generated by sending an $n_A$-dimensional quantum system from Alice to Bob, that is,
\begin{equation}
    q(b|xy)=\operatorname{Tr}\!\left(M_{b|y}\rho_x\right),
\end{equation}
where $\rho_x$ are states on $\mathbb{C}^{n_A}$ and $\{M_{b|y}\}_b$ is a POVM for each $y$.

\begin{Theorem}\label{thm:quantum_subset_nonadaptive_NS}
For every PM scenario $(n_X;n_Y,n_B)$ and every $n_A$, one has
\begin{equation}
    \mathrm{Q}(n_X;n_Y,n_B)_{n_A}
    \subseteq
    \mathrm{NS}^{\mathrm{NA}}(n_X;n_Y,n_B)_{n_A}.
\end{equation}
\end{Theorem}

\begin{proof}
Let $\mathbf{q}\in \mathrm{Q}(n_X;n_Y,n_B)_{n_A}$. 
For each fixed value of Bob's input $y$, define the single-input sub-behavior
\begin{equation}
    q_y(b|x):=q(b|xy).
\end{equation}
This sub-behavior is generated by sending an $n_A$-dimensional quantum system in a PM scenario where Bob has no inputs. 
By the seminal result of Frenkel and Weiner~\cite{frenkel-storage-2015}, every such quantum behavior admits a classical model with an $n_A$-valued message. 
Therefore,
\begin{equation}
    \mathbf{q}_y\in \mathrm{C}(n_X;1,n_B)_{n_A}
    \quad
    \forall y\in[n_Y].
\end{equation}
The claim follows directly from \cref{res:NS_equals_C}.
\end{proof}

This inclusion has a direct operational interpretation in terms of the simulation of quantum communication. 
In a single use of an $n_A$-dimensional quantum channel, the most general procedure is the following: Alice prepares an arbitrary state $\rho$ on $\mathbb{C}^{n_A}$, sends it to Bob, and Bob performs an arbitrary POVM $\mathsf{M}=\{B_b\}_{b=0}^{n_B-1}$ on the received system. 
The accessible statistics are exactly those given by Born's rule,
\begin{equation}
    p_Q(b|\rho,\mathsf{M})
    =
    \operatorname{Tr}(\rho B_b).
\end{equation}
Thus, simulating the communication of an $n_A$-dimensional quantum system means reproducing these conditional probabilities for arbitrary choices of the state $\rho$ and the measurement $\mathsf{M}$.

Since \cref{thm:quantum_subset_nonadaptive_NS} holds for every PM scenario, it applies to arbitrary finite collections of states and measurements. 
By a standard compactness argument, this finite-scenario simulation implies a simulation for arbitrary choices of $\rho$ and $\mathsf{M}$. 
We can summarize this operational consequence as follows.

\begin{result}\label{cor:qudit_simulation_NS}
The transmission of an $n_A$-dimensional quantum system can be simulated by the transmission of an $n_A$-valued classical message assisted non-adaptively by arbitrary NS correlations.
\end{result}
In this sense, non-adaptive NS assistance collapses qudit communication into classical dit communication of the same alphabet size. 

This should be contrasted with more restricted resources. 
Without NS assistance, even in the presence of shared randomness, the exact classical simulation of quantum communication is a much more demanding problem: the transmission of a qubit is known to be simulable with two classical bits, and this cost is optimal~\cite{toner_communication_2003,renner_classical_2023}. 
For qutrits and higher-dimensional systems, it remains open whether an exact simulation is possible with any finite amount of classical communication~\cite{schlosser_bounding_2026}. 
Similarly, entanglement assistance together with an $n_A$-valued classical message does not, in general, reproduce all PM behaviors generated by sending an $n_A$-dimensional quantum system \cite{tavakoli2021correlations, Entanglement_Marcin2010}. 
Thus, non-adaptive NS assistance provides a particularly strong relaxation of quantum communication.

A second consequence of \cref{res:NS_equals_C} follows from simple regimes in which the reduced classical set $\mathrm{C}(n_X;1,n_B)_{n_A}$ already coincides with the full set of valid behaviors. 
In these cases, the characterization of \cref{res:NS_equals_C} implies that non-adaptive NS assistance trivializes the corresponding PM scenario. Indeed, let us recall first that, whenever $n_A \geq n_X$, the classical PM scenario trivializes: Alice can communicate her input $x$ directly to Bob, and arbitrary behaviors can be generated.  
There is another simple regime in which the reduced classical set $\mathrm{C}(n_X;1,n_B)_{n_A}$ trivializes: when $n_A \geq n_B$. 
In this case, since Bob has no inputs, Alice can compute and send the outcome $b$ that Bob is required to output, enabling the construction of arbitrary behaviors $\mathbf{q}=[q(b|x)]$.
By \cref{res:NS_equals_C}, it follows that in arbitrary PM scenarios $(n_X;n_Y,n_B)$ one can generate all valid behaviors using a classical message of dimension $n_A \geq n_B$ with non-adaptive NS assistance. That is,
\begin{equation}\label{eq:Trivializantion_NS_NA}
\mathrm{NS}^{\mathrm{NA}}(n_X;n_Y,n_B)_{n_A \ge n_B}
= \big\{\mathbf{q} \:\big|\: q(b|xy) \ge 0,\; \sum_{b} q(b|xy) = 1\big\}.
\end{equation}

The content of \cref{eq:Trivializantion_NS_NA} can be interpreted as a trivialization of the PM scenario by non-adaptive non-signaling assistance: whenever $n_A \ge n_B$, any communication task can be achieved with maximal algebraic success probability. In particular, this applies to several prominent families of PM scenarios, including various random access codes—for instance, the seminal $2 \to 1$ RAC \cite{ambainis_dense_1999, Entanglement_Marcin2010}, corresponding to $(4; 2, 2)$ with the communication of a bit, as well as the more general $n \to d$ RAC, corresponding to $(2^n; n, 2)$ with the communication of a dit \cite{Ambainis:2019QIP, de_Vicente_2019}.

Moreover, \cref{eq:Trivializantion_NS_NA} can be viewed as an alternative and independent proof of van Dam’s seminal result on the collapse of communication complexity~\cite{van_dam_implausible_2013}; notably, it relies solely on simple structural properties of the classical and non-adaptive NS sets, rather than the explicit constructions adopted in van Dam's original proof. 
\begin{remark}[Collapse of Communication Complexity~\cite{van_dam_implausible_2013}]
    Consider a PM scenario where Alice and Bob aim to compute a function $f: [n_X] \times [n_Y] \rightarrow [n_B]$ of their respective inputs $x \in [n_X]$ and $y \in [n_Y]$, with minimal communication.
    More specifically, for a given function $f$, the object of interest is the minimal dimension $n_A$ of the classical message sent by Alice to Bob such that they generate the behavior $q(b \vert x y) = \delta_{b, f(x, y)}$.
    
    According to \cref{eq:Trivializantion_NS_NA}, any such function $f$ can be computed with the communication of an $n_B$-dimensional classical message assisted non-adaptively by NS correlations.
    As a particular case, any Boolean function $f: [n_X] \times [n_Y] \rightarrow \{0,1\}$ can be computed with the communication of a single bit.
\end{remark}

Outside the regimes $n_A \geq n_B$ and $n_A \geq n_X$, the classical polytopes $\mathrm{C}(n_X;1,n_B)_{n_A}$ no longer coincide with the full probability set, and additional constraints are needed to characterize them. Nevertheless, these polytopes are closely related to a well-studied problem in information theory: the simulation of classical channels with limited communication.

Indeed, a behavior $\mathbf{q}=[q(b|x)]$ in the scenario $(n_X;1,n_B)$ can be viewed as a classical channel from an input alphabet of size $n_X$ to an output alphabet of size $n_B$, where $q(b|x)$ is the probability that the channel outputs $b \in [n_B]$ upon receiving input $x \in [n_X]$~\cite{Cover2006ElementsOfInfoTheory}. 
In this language, understanding $\mathrm{C}(n_X;1,n_B)_{n_A}$ amounts to asking which classical channels can be simulated using a message of alphabet size $n_A$, together with shared randomness. 
This channel-simulation perspective has been extensively studied in the literature~\cite{Bennett2001ReverseShannon,Fang2018ChannelSimulation, Wang2016QuantumNSAssisted, Duan2016NSAssisted,Berta2013EntanglementCost,Berta2011QuantumReverseShannonWithOne-Shot, Doolittle_2021_Certifying, Cubitt:2011TIT, Cubitt:2010PRL, Chitambar2023CommunicationValueOfQuantumChannel}.

\subsection{A characterization of the adaptive NS-assisted set of behaviors for PM scenarios}\label{sub_sec:B}
%%%%%
We now turn our attention to the set of behaviors that admit an adaptive NS-assisted model.
Unlike the non-adaptive case, whose characterization relies on the classical sets of scenarios where Bob has no inputs (see \cref{res:NS_equals_C}), the adaptive set 
$\mathrm{NS}^{\mathrm{A}}(n_X; n_Y, n_B)_{n_A}$ 
can be described directly by a simple family of inequalities.
These inequalities admit a simple information-theoretic interpretation.

A natural way to understand them is through input-discrimination tasks. 
Let us first consider the canonical case, in which Bob is asked to output the label of Alice's input itself. 
This corresponds to the PM scenario with $n_Y=1$ and $n_B=n_X$, where a successful guess is described by $b=x$. 
To make the task non-trivial, we assume that Alice's message alphabet is smaller than her input alphabet, $n_A < n_X$. 
The performance of a behavior $\mathbf{q}$ in this discrimination task is quantified by the total success score
\begin{equation}
    \mathcal{S}_{\mathrm{suc}}(\mathbf{q}) = \sum_{b=0}^{n_B-1} q(b|x=b).
\end{equation}

With a classical message of alphabet size $n_A$, the total success score is at most $n_A$~\cite{gallego2010pam}. 
Operationally, this bound reflects the fact that the message can perfectly identify at most $n_A$ distinct inputs of Alice. 
The bound is tight: Alice can choose $n_A$ inputs, say $x \in [n_A]$, and transmit them perfectly by sending $a=x$. 
For all remaining inputs $x \geq n_A$, she sends a fixed default message, say $a=0$. 
If Bob outputs the received message, $b=a$, the resulting behavior satisfies $q(b|x)=\delta_{b,x}$ for $x<n_A$ and $q(b|x)=\delta_{b,0}$ for $x\geq n_A$, giving
\begin{equation}
    \mathcal{S}_{\mathrm{suc}}(\mathbf{q})  = \sum_{b=0}^{n_A-1}1 + \sum_{b=n_A}^{n_B-1}0  =  n_A.
\end{equation}

It is then natural to ask whether adaptive NS assistance can improve the performance of this input-discrimination task. 
The answer is negative: even with adaptive NS assistance, the total success score is still bounded by $n_A$. 
Indeed, if the behavior $\mathbf{q}=[q(b|x)]$ can be realized with a $n_A$-dimensional classical message and adaptive NS assistance, then
\begin{equation}\label{eq:GuessInputTask}
    \begin{aligned}
        \sum_b q(b|x=b)
        &= \sum_b \sum_a p(a,b|x=b,z=a) \\
        &\leq \sum_b \sum_a \sum_{\tilde{a}} p(\tilde{a},b|x=b,z=a) \\
        &= \sum_a \sum_b p_B(b|z=a) \\
        &= \sum_a 1 = n_A,
    \end{aligned}
\end{equation}
where $p_B$ denotes Bob's marginal distribution. 
Thus, the same bottleneck imposed by the size of the classical message remains in place for adaptive NS assistance.

This bound can also be viewed as a structural constraint on the adaptive NS-assisted set 
$\mathrm{NS}^{\mathrm{A}}(n_X;1,n_X)_{n_A}$. 
For instance, the behavior corresponding to perfect input discrimination, $q(b|x)=\delta_{b,x}$, has success score 
$\mathcal{S}_{\mathrm{suc}}(\mathbf{q})=n_X$. 
Since we are assuming $n_X>n_A$, this behavior violates the bound in \cref{eq:GuessInputTask}. 
Hence, it does not belong to 
$\mathrm{NS}^{\mathrm{A}}(n_X;1,n_X)_{n_A}$.

The previous task corresponds to the canonical input-discrimination setting in which Bob is asked to output Alice's input: for each input $x$, the only correct output is $b=x$. 
More generally, one may consider discrimination tasks in which Bob is only required to distinguish a selected family of Alice's inputs, and where the possible answers need not be labeled by the inputs themselves. 
Such a discrimination problem can be encoded by a function $f:[n_B]\to[n_X]$, which specifies which input of Alice is identified by each output of Bob: the output $b$ is declared successful precisely when Alice's input is $x=f(b)$. 
Equivalently, for each $x$ in the image of $f$, the accepted outputs are those in $f^{-1}(x)$. 
The corresponding score is
\begin{equation}
    \sum_b q(b|x=f(b)).
\end{equation}
Thus, although success is no longer defined by the canonical rule $b=x$, the score still measures Bob's ability to discriminate the inputs in the image of $f$, allowing several outputs to certify the same input.

The same argument used in \cref{eq:GuessInputTask} applies to any such discrimination task. 
Indeed, if $\mathbf{q}=[q(b|x)]$ can be realized with a $n_A$-dimensional classical message and adaptive NS assistance, then for every function $f:[n_B]\to[n_X]$ one has
\begin{equation}\label{eq:f_inequalities_no_y}
    \sum_b q(b|f(b)) \leq n_A .
\end{equation}
This family of inequalities admits an equivalent compact form. 
Indeed, imposing \cref{eq:f_inequalities_no_y} for all functions $f:[n_B]\to[n_X]$ is equivalent to imposing
\begin{equation}\label{eq:max_constraint_no_y}
    \sum_{b=0}^{n_B-1}\max_{x\in[n_X]} q(b|x) \leq n_A .
\end{equation}
In other words, the left-hand side of \cref{eq:max_constraint_no_y} is the largest input-discrimination score that can be extracted from the behavior $\mathbf{q}$.

Interestingly, the quantity on the left-hand side of \cref{eq:max_constraint_no_y}, has appeared in several contexts in the literature.
In channel simulation, it appears as a maximum-likelihood-type quantity~\cite{Doolittle_2021_Certifying} and as the communication value of a classical channel~\cite{Chitambar2023CommunicationValueOfQuantumChannel}. 
In quantitative information flow, it is closely related to the multiplicative Bayes capacity~\cite{Alvim2020BookQuantitativeInfoFlow, Alvim2012MeasuringInfoLeakage, Braun2009QuantitativeNotions}. 
It also appears naturally in classical state-discrimination problems~\cite{Kvashchuk2026MostDiscriminableStates}. 
In all these settings, it quantifies, in one form or another, how well one can infer Alice's input from Bob's output.

In an arbitrary PM scenario $(n_X;n_Y,n_B)$, the same reasoning can be applied to each fixed value of Bob's input. 
Hence, every adaptive NS-assisted behavior $\mathbf{q}=[q(b|xy)]$ must satisfy
\begin{equation}\label{eq:f_inequalities}
    \sum_b q(b|f(b),y) \leq n_A
\end{equation}
for every $y\in[n_Y]$ and every function $f:[n_B]\to[n_X]$, or equivalently,
\begin{equation}\label{eq:max_constraint_adaptive}
    \sum_{b=0}^{n_B-1}\max_{x\in[n_X]} q(b|x,y) \leq n_A
    \quad
    \forall y\in[n_Y].
\end{equation}

It is worth commenting that the same family of constraints has appeared previously in the study of theory-independent models with informational restrictions~\cite{tavakoli_informationally_2020, chaturvedi_quantum_2020, tavakoli_informationally_2022, pauwels_information_2025, pandit_limits_2026}. 
In that context, one considers the set of all valid PM behaviors satisfying positivity, normalization, and the information constraints in \cref{eq:max_constraint_adaptive}. 
These models have been used as theory-independent outer approximations to PM behaviors generated either with quantum communication or with entanglement-assisted classical communication, under the same message-size restriction $n_A$. 
We denote this set by $\mathcal{G}(n_X;n_Y,n_B)_{n_A}$.

The derivation above shows that every adaptive NS-assisted behavior belongs to this set, namely
\begin{equation}
    \mathrm{NS}^{\mathrm{A}}(n_X;n_Y,n_B)_{n_A}
    \subseteq
    \mathcal{G}(n_X;n_Y,n_B)_{n_A}.
\end{equation}
Our main result in this subsection shows that this inclusion is tight: the input-discrimination constraints are not only necessary for adaptive NS assistance, but also sufficient.

\begin{result}\label{res:NS_equals_G}
For any PM scenario $(n_X;n_Y,n_B)$, a behavior $\mathbf{q}=[q(b|xy)]$ can be realized with an $n_A$-dimensional classical message and adaptive NS assistance if and only if it satisfies
\begin{equation}\label{eq:max_teo}
    \sum_{b=0}^{n_B-1}\max_{x\in[n_X]} q(b|x,y) \leq n_A
    \quad
    \forall y\in[n_Y].
\end{equation}

In other words, the set of adaptive NS-assisted behaviours is equal to the set of behaviours only restricted by input-discrimination constraints
\begin{equation}
    \mathrm{NS}^{\mathrm{A}}(n_X;n_Y,n_B)_{n_A}
    =
    \mathcal{G}(n_X;n_Y,n_B)_{n_A}.
\end{equation}
\end{result}

The proof of \cref{res:NS_equals_G} has two ingredients. 
The first is an adaptive analogue of \cref{thm:Porto_thm}: adaptive NS assistance in an arbitrary PM scenario can be reduced to adaptive NS assistance in the corresponding scenarios where Bob has no inputs. 

\begin{restatable}[]{Theorem}{thmadaptivelocalglobal}\label{global_local_adaptive}
Let $\mathbf{q} = [q(b|xy)]$ be a behavior of the PM scenario $(n_X; n_Y, n_B)$, and let $\mathbf{q}_y$ denote the behaviors of the scenario $(n_X; 1, n_B)$ constructed via $q_{y}(b|x) := q(b|xy)$. 
    Then, $\mathbf{q} \in \mathrm{NS}^{\mathrm{A}}(n_{X}; n_{Y}, n_{B})_{n_{A}}$ if, and only if, $\mathbf{q}_y \in \mathrm{NS}^{\mathrm{A}}(n_{X}; 1, n_{B})_{n_{A}}$ for all $y \in [n_y]$.

    In words, a behavior admits an adaptive NS-assisted model in the full PM scenario if and only if each of all its fixed-$y$ sub-behaviors also admits an adaptive NS-assisted model.
\end{restatable}

The proof of \cref{global_local_adaptive} follows the same general strategy as the proof of \cref{thm:Porto_thm}, and is deferred to Appendix~\ref{appdx:proof_global_local_adaptive}. 
A central ingredient is \cref{lemma:uniform_adaptive}, which shows that whenever a behavior admits an adaptive NS-assisted model, one can construct an equivalent adaptive model in which Alice's marginal distribution in the assisting NS box is uniform. 
This is the adaptive analogue of the uniform-marginal statement in \cref{lemma:NA}.

The second ingredient is a known result on the simulation of classical channels with NS assistance. 
In light of \cref{global_local_adaptive}, it is enough to characterize adaptive NS assistance in PM scenarios where Bob has no inputs. 
As discussed at the end of the previous subsection, such scenarios are directly related to the problem of simulating classical channels with limited communication. 
For this problem, NS-assisted simulation was studied in \cite{Cubitt:2011TIT}, where a simple characterization was obtained. 
Translated into our notation, \cite[Theorem~16]{Cubitt:2011TIT} reads as follows.

\begin{restatable}[Theorem~16 of~\cite{Cubitt:2011TIT}]{Theorem}{thmCubitt}\label{thm:adaptive_NS_no_inputs}
A behavior $\mathbf{q}=[q(b|x)]$ of a PM scenario $(n_X;1,n_B)$ belongs to 
$\mathrm{NS}^{\mathrm{A}}(n_X;1,n_B)_{n_A}$ if and only if
\begin{equation}\label{eq:n_A_bound_on_cv}
    \sum_{b=0}^{n_B-1}\max_{x\in[n_X]} q(b|x) \leq n_A.
\end{equation}
\end{restatable}

We can now conclude the proof of \cref{res:NS_equals_G}. 
By \cref{global_local_adaptive}, a behavior $\mathbf{q}=[q(b|xy)]$ belongs to 
$\mathrm{NS}^{\mathrm{A}}(n_X;n_Y,n_B)_{n_A}$ if and only if, for every $y\in[n_Y]$, the reduced behavior $\mathbf{q}_y$ belongs to 
$\mathrm{NS}^{\mathrm{A}}(n_X;1,n_B)_{n_A}$. 
Using \cref{thm:adaptive_NS_no_inputs}, $\mathbf{q}_y$ belongs to 
$\mathrm{NS}^{\mathrm{A}}(n_X;1,n_B)_{n_A}$ if and only if,
\begin{equation}
    \sum_{b=0}^{n_B-1}\max_{x\in[n_X]} q_y(b|x) \leq n_A
    \quad
    \forall y\in[n_Y].
\end{equation}
As $q_y(b|x) := q(b|xy)$, this proves \cref{res:NS_equals_G}.

We close this subsection by pointing out some direct consequences of \cref{res:NS_equals_G}. First, the characterization gives a simple membership test for adaptive NS-assisted behaviors. 
Given a valid PM behavior $\mathbf{q}=[q(b|xy)]$, deciding whether $\mathbf{q}$ belongs to $\mathrm{NS}^{\mathrm{A}}(n_X;n_Y,n_B)_{n_A}$ only requires checking the inequalities of \cref{eq:max_teo}.
Thus, membership can be decided by a direct evaluation of the input-discrimination score associated with each fixed value of Bob's input.

Second, the same result gives an explicit half-space representation of the adaptive NS-assisted set, and therefore an alternative linear-programming formulation for optimizing linear PM functionals over it.
Appendix~\ref{appdx:hyperplane_cci} discusses a succinct representation of these constraints.
Moreover, this half-space description makes it possible to apply standard polyhedral tools, such as PANDA~\cite{panda}, to obtain an equivalent description in terms of extremal points.
In this sense, one can identify the extremal adaptive NS-assisted PM behaviors.

\subsection{Adaptive advantage of NS assistance in PM scenarios}
At this point, we have characterized both the non-adaptive and adaptive NS-assisted sets of behaviors in arbitrary PM scenarios; see \cref{res:NS_equals_C,res:NS_equals_G}. 
We now compare these sets in order to determine when adaptivity provides a strict advantage. Notably, we identify all PM scenarios for which an adaptive advantage of NS assistance exists over its non-adaptive use.
As suggested by the results of the previous sections, adaptive advantages for NS assistance can be understood entirely in terms of fixed-$y$ sub-behaviors, or equivalently in terms of the corresponding PM scenarios where Bob has no inputs.
We therefore begin with the single-input case.

In PM scenarios where Bob has no inputs, non-adaptive NS assistance coincides with the classical set; see \cref{remark_NSeqClassical_Yeq1}. 
Therefore, any non-classical behavior in such a scenario that can be generated with adaptive NS assistance already witnesses an adaptive advantage over non-adaptive NS assistance. 
This situation occurs, for example, in scenarios studied in the context of entanglement-assisted communication~\cite{Frenkel2022entanglement, vieira_interplays_2023, manna_entanglement_2026}. 
Together with known results on PM inequalities and communication-constrained sets~\cite{vieira_interplays_2023,Doolittle_2021_Certifying}, our characterization of adaptive NS assistance leads to a complete classification of when post-classical adaptive behaviors exist in the single-input case.

\begin{Theorem}\label{thm:Y1_classification}
For $n_A\geq 2$, the adaptive NS-assisted set 
$\mathrm{NS}^{\mathrm{A}}(n_X;1,n_B)_{n_A}$ contains non-classical behaviors if and only if
\begin{equation}
    n_X \geq n_A+1
    \quad\text{and}\quad
    n_B \geq n_A+2.
\end{equation}
Equivalently, these are exactly the PM scenarios with $n_Y=1$ in which adaptive NS assistance provides an advantage over its non-adaptive use.
\end{Theorem}

We prove the two directions separately. 
First, we show that the above conditions are sufficient.
Here and throughout this section, we use $(n_X;n_Y,n_B)_{n_A}$ to denote the PM scenario $(n_X;n_Y,n_B)$ equipped with a classical message of alphabet size $n_A$, while the allowed assistance is specified separately in each case.

\begin{lemma}\label{lemma:Y1_sufficiency}
For $n_A\geq 2$, any PM scenario $(n_X;1,n_B)_{n_A}$ satisfying $n_X \geq n_A+1$ and $n_B \geq n_A+2$
admits an adaptive NS-assisted behavior that is not classical.
\end{lemma}

\begin{proof}
First consider the minimal scenario $(n_A+1;1,n_A+2)_{n_A}$. 
It was shown in~\cite{vieira_interplays_2023} that every classical behavior in this scenario satisfies the PM inequality
\begin{equation}\label{eq:interplays_facet_general}
    \mathcal{F}_{n_A}(\mathbf{q}) = 2\sum_{i=0}^{n_A} q(i|i) + \sum_{i=0}^{n_A} q(n_A+1|i) \leq 2n_A .
\end{equation}
Now consider the behavior
\begin{equation}\label{eq:adaptive_NS_family}
    q_\star(b|x) = \frac{n_A-1}{n_A}\delta_{b,x} + \frac{1}{n_A}\delta_{b,n_A+1}.
\end{equation}
This behavior belongs to the adaptive NS-assisted set. 
Indeed, using \cref{eq:max_constraint_no_y}, we have
\begin{equation}
    \sum_{b=0}^{n_A+1}\max_{x\in[n_A+1]} q_\star(b|x) = \sum_{b=0}^{n_A}\frac{n_A-1}{n_A} + \frac{1}{n_A} = n_A.
\end{equation}
Thus, $\mathbf{q}_\star\in\mathrm{NS}^{\mathrm{A}}(n_A+1;1,n_A+2)_{n_A}$.

On the other hand, evaluating \cref{eq:interplays_facet_general} on $\mathbf{q}_\star$ gives
\begin{equation}
    \mathcal{F}_{n_A}(\mathbf{q}_\star) = 2(n_A+1)\frac{n_A-1}{n_A} + (n_A+1)\frac{1}{n_A} = 2n_A+\frac{n_A-1}{n_A} > 2n_A,
\end{equation}
for every $n_A\geq 2$. 
Therefore, $\mathbf{q}_\star$ is not classical.

Finally, the same conclusion holds for every scenario with $n_X\geq n_A+1$ and $n_B\geq n_A+2$. 
Indeed, $\mathbf{q}_\star$ can be embedded into the larger scenario by assigning zero probability to the additional outputs and by repeating one of the original input distributions for the additional inputs. 
This embedding preserves membership in the adaptive NS-assisted set, since the value of the quantity in \cref{eq:max_constraint_no_y} does not increase. 
It also preserves non-classicality, because restricting the larger behavior to the original inputs and outputs recovers a behavior that violates \cref{eq:interplays_facet_general}. 
\end{proof}

We now turn to the converse direction of \cref{thm:Y1_classification}. 
Suppose that a PM scenario $(n_X;1,n_B)_{n_A}$ does not satisfy the conditions in \cref{lemma:Y1_sufficiency}. 
There are two immediate cases. 
If $n_A\geq n_X$, Alice can communicate her input directly to Bob, and the classical set already coincides with the set of all valid behaviors. 
If $n_A\geq n_B$, then, since Bob has no input, Alice can sample the desired output and communicate it to Bob; again, the classical set coincides with the full behavior set. 
In both cases, adaptive NS assistance cannot generate any post-classical behavior.

It remains to consider the only nontrivial family not covered by \cref{lemma:Y1_sufficiency}, namely
\begin{equation}
    (n_X;1,n_A+1)_{n_A},
    \qquad
    n_X\geq n_A+1.
\end{equation}
This case was settled in \cite{Doolittle_2021_Certifying}. 
Translated into our notation, their result reads as follows.

\begin{restatable}[Theorem~1 of~\cite{Doolittle_2021_Certifying}]{Theorem}{thmDoolittle}\label{thm:Doolittle_Y1}
For any PM scenario $(n_X;1,n_A+1)_{n_A}$ with $n_X\geq n_A+1$, a behavior $\mathbf{q}=[q(b|x)]$ is classical if and only if
\begin{equation}\label{eq:Doolittle_condition}
    \sum_{b=0}^{n_A}\max_{x\in[n_X]} q(b|x) \leq n_A.
\end{equation}
\end{restatable}

On the other hand, by \cref{res:NS_equals_G}, the same condition characterizes the adaptive NS-assisted set in this scenario. 
Therefore,
\begin{equation}
    \mathrm{NS}^{\mathrm{A}}(n_X;1,n_A+1)_{n_A}
    =
    \mathrm{C}(n_X;1,n_A+1)_{n_A}.
\end{equation}
Hence, adaptive NS assistance produces no post-classical behaviors in this remaining family. 
Together with \cref{lemma:Y1_sufficiency}, this proves \cref{thm:Y1_classification}.

Having classified the single-input case, we now explain how such advantages extend to PM scenarios in which Bob has several inputs. 
The construction is simple and shows that these advantages still originate from the corresponding $n_Y=1$ scenario.

Let $\mathbf{\tilde{q}}=[\tilde{q}(b|x)]$ be a post-classical behavior in $(n_X;1,n_B)_{n_A}$ that admits an adaptive NS-assisted model. 
For any $n_Y>1$, define a behavior in the larger scenario $(n_X;n_Y,n_B)_{n_A}$ by assigning the same single-input behavior to every fixed-$y$ component:
\begin{equation}
    \tilde{q}(b|xy) = \tilde{q}(b|x)
    \quad
    \forall y\in[n_Y].
\end{equation}

By \cref{global_local_adaptive}, the constructed behavior admits an adaptive NS-assisted model in $(n_X;n_Y,n_B)_{n_A}$. On the other hand, by \cref{res:NS_equals_C}, it does not admit a non-adaptive NS-assisted model in $(n_X; n_Y, n_B)_{n_A}$, as its sub-behaviors are non-classical\footnote{An alternative argument leads to the same conclusion using fewer than $n_{Y}$ copies of the post-classical behavior $\mathbf{\tilde{q}}$. In particular, it suffices to take a single copy—for instance, setting $q(b|x,y=0)=\tilde{q}(b|x)$ and assigning classical behaviors to the remaining components.}. Therefore, $q(b|xy)$ exhibits an adaptive advantage over non-adaptive NS assistance.

The construction above shows that any adaptive advantage present in a single-input scenario can be lifted to scenarios with arbitrary $n_Y$ by assigning suitable fixed-$y$ sub-behaviors. 
However, such an advantage is not genuinely tied to the presence of several inputs for Bob: it is inherited from the corresponding single-input component. 
The following corollary shows that this is not merely a feature of the construction, but a general property of NS assistance.

\begin{corollary}\label{cor:main}
A PM behavior $\mathbf{q} = [q(b|xy)]$ exhibits an advantage of adaptive over non-adaptive NS assistance if, and only if, there exists at least one input $y' \in [n_{Y}]$ such that the sub-behavior $\mathbf{q}_{y'} \equiv [q(b|x y')]_{b,x}$ also exhibits this advantage.

More formally, given $\mathbf{q} \in \mathrm{NS}^{\mathrm{A}}(n_{X}; n_{Y}, n_{B})_{n_{A}}$, then $\mathbf{q} \notin \mathrm{NS}^{\mathrm{NA}}(n_{X}; n_{Y}, n_{B})_{n_{A}}$ if, and only if, there exists an input $y'$ such that the sub-behavior satisfies $\mathbf{q}_{y'} \in \mathrm{NS}^{\mathrm{A}}(n_{X}; 1, n_{B})_{n_{A}}$ but $\mathbf{q}_{y'} \notin \mathrm{NS}^{\mathrm{NA}}(n_{X}; 1, n_{B})_{n_{A}}$.
\end{corollary}

The proof follows directly from \cref{global_local_adaptive} together with the contrapositive of \cref{thm:Porto_thm}.

\Cref{cor:main} shows that, for NS assistance, adaptive advantages cannot arise from a genuinely multi-input mechanism. 
They are always witnessed already at the level of at least one single-input sub-behavior. 
This is in sharp contrast with the entanglement-assisted case. 
In~\cite{adaptive2022pauwels}, the authors constructed a behavior in the scenario $(4;2,2)_3$ that admits an adaptive EA model but no non-adaptive EA model. 
Such an advantage cannot be traced back to the corresponding single-input scenario $(4;1,2)_3$: in that case, the classical set already coincides with the full set of valid behaviors, and hence no separation between adaptive and non-adaptive EA strategies is possible. 
The advantage in~\cite{adaptive2022pauwels} is therefore genuinely tied to the presence of multiple inputs for Bob.

At the level of whole PM scenarios, \cref{cor:main} implies that a scenario $(n_X;n_Y,n_B)_{n_A}$ exhibits an adaptive advantage for NS assistance if and only if the corresponding single-input scenario $(n_X;1,n_B)_{n_A}$ does. 
Combining this observation with \cref{thm:Y1_classification} gives the desired classification.

\begin{result}\label{res:classification}
For $n_A\geq 2$, a PM scenario $(n_X;n_Y,n_B)_{n_A}$ exhibits an adaptive advantage of NS assistance over its non-adaptive use if and only if
\begin{equation}
    n_X \geq n_A+1
    \quad\text{and}\quad
    n_B \geq n_A+2,
\end{equation}
for any $n_Y\geq 1$.
\end{result}

\cref{res:classification} completes the comparison between adaptive and non-adaptive NS assistance at the level of PM scenarios: a strict separation occurs exactly when $n_X\geq n_A+1$ and $n_B\geq n_A+2$. 

\section{Conclusion}\label{sec:conclusions}
In this work, we investigated non-signaling (NS) assistance in prepare-and-measure (PM) scenarios with classical communication, encompassing both adaptive and non-adaptive protocols.
We showed that, for arbitrary PM scenarios, the non-adaptive set of NS-assisted behaviors can be fully characterized in terms of the classical set in the corresponding scenario where Bob has no inputs (\cref{res:NS_equals_C}).
As a consequence, combining this characterization with the result of Frenkel and Weiner~\cite{frenkel-storage-2015}, we showed that non-adaptive NS assistance is already strong enough to reproduce quantum communication with the same message dimension: an $n_A$-dimensional quantum system can be simulated by an $n_A$-valued classical message assisted by arbitrary NS correlations (\cref{thm:quantum_subset_nonadaptive_NS,cor:qudit_simulation_NS}).
We then established an analogous characterization for the adaptive set: a behavior admits an adaptive NS-assisted model if and only if it satisfies a simple family of input-discrimination inequalities associated with the size of Alice's message (\cref{res:NS_equals_G}). 
These characterizations are obtained by reducing arbitrary PM scenarios to single-input subscenarios and by connecting the resulting problems to channel-simulation, or reverse-Shannon-type, questions; in particular, our adaptive characterization builds on the result of Ref.~\cite{Cubitt:2011TIT}, translated and combined with our reduction framework. 
The same input-discrimination constraints have also appeared in the literature on \textit{theory-independent models with informational restrictions}~\cite{tavakoli_informationally_2020, tavakoli_informationally_2022}; here, we identify them as the exact operational limitations imposed by adaptive NS assistance. 
Finally, combining our characterization with PM inequalities from Ref.~\cite{vieira_interplays_2023} and the single-input communication-constrained characterization of Ref.~\cite{Doolittle_2021_Certifying}, we determined all PM scenarios in which a strict advantage of adaptive NS assistance over non-adaptive protocols exists (\cref{res:classification}).

As for future lines of investigation, our results suggest some directions. 
One natural possibility is to use NS assistance as a tool to study semi-device-independent communication protocols from a theory-independent perspective. 
In Bell scenarios, the NS principle has played an important role in understanding the limitations of device-independent protocols, including cryptographic and randomness-generation tasks. 
Our framework provides an analogue of this approach in PM scenarios, where communication constraints replace the role of measurement locality. 
This may be useful, for instance, for analyzing the security or limitations of PM-based cryptographic protocols against adversaries constrained only by non-signaling principles.

Another possible direction is to use our characterizations as general simulation criteria for other PM resource models. 
Indeed, \cref{res:NS_equals_C} shows that any resource whose single-input behaviors are always classical with a message of size $n_A$ is automatically contained in the non-adaptive NS-assisted classical model with the same message size. 
Our quantum-communication simulation result is one instance of this principle, using the theorem of Frenkel and Weiner \cite{frenkel-storage-2015}. 
Similarly, \cref{res:NS_equals_G} shows that any resource whose behaviors satisfy the input-discrimination constraints of \cref{eq:max_teo} is simulated in the adaptive NS-assisted model. 
It would be interesting to investigate what other nontrivial simulation results, inclusions, or separations between PM resource models can be obtained from these two general facts.

\begin{acknowledgments}
We thank Marco Túlio Quintino for useful discussions.
JN is grateful to the University of Siegen for the welcoming environment in which the groundwork for this project was developed. This study was financed in part by the Coordenação de Aperfeiçoamento de Pessoal de Nível Superior - Brasil (CAPES) - Finance Code 001 - and by the Brazilian National Council for Scientific and Technological Development (CNPq) (INCT-IQ and Grant Numbers 316657/2023-9 and 200704/2025-7). This research was also supported by the São Paulo Research Foundation (Fapesp) under grant no.~2024/16657-3, 2025/01058-0 and 2025/27304-7. This work is partially carried out under IRA Programme, project no. FENG.02.01-IP.05-0006/23, financed by the FENG program 2021-2027, Priority FENG.02, Measure FENG.02.01., with the support of the FNP.  This work was supported by the Conselho Nacional de Desenvolvimento Científico e Tecnológico (CNPq), through a grant from the Conhecimento Brasil Program - Line 1.
In addition, this work was supported by the Deutsche Forschungsgemeinschaft 
(DFG, German Research Foundation, project number 563437167), the 
Sino-German Center for Research Promotion (Project M-0294),
and the German Federal Ministry of Research, Technology and Space (Project QuKuK, Grant No. 16KIS1618K and Project BeRyQC, Grant No. 13N17292).
This work was funded by QuantEdu France, a State aid managed by the French National Research Agency for France 2030 with the reference ANR-22-CMAS-0001.

\end{acknowledgments} 

\bibliography{1_citations}

%%%%%%%%%%%%%%%%%%%%%%%%%%%%%%%%%%%%%%%%%%%%%%%%%%%%%%%%%%%%%%%%%%%%%%%%%%%%%%%%%%%%%%%%%%
\appendix
%%%%%%%%%%%%%%%%%%%%%%%%%%%%%%%%%%%%%%%%%%%%%%%%%%%%%%%%%%%%%%%%%%%%%%%%%%%%%%%%%%%%%%%%%%

\setcounter{equation}{0}
\renewcommand{\theequation}{\thesection\arabic{equation}}
\renewcommand{\theHequation}{\thesection.\arabic{equation}}

\section{A proof for \cref{lemma:NA}}\label{appdx:proofPropo}
\LemmaNA*
\begin{proof}
    Since \cref{eq:NA_NS_model_2} is a particular case of \cref{eq:NA_NS_model_1}, to prove the equivalence between them we just have to prove that the latter implies the former.
    If a behavior $\mathbf{q}$ admits a decomposition as in \cref{eq:NA_NS_model_1}, define
    \begin{equation}
    \begin{aligned}
        p(a; b_0 \ldots b_{n_A - 1}\vert x; y) := & \sum_\beta p(a \beta \vert xy) \\
        & \times p(b_0\vert 0, \beta) \ldots p(b_{n_A - 1} \vert n_A-1, \beta).
    \end{aligned}
    \end{equation}
    With a bipartite box $p(a; b_0 \ldots b_{n_A - 1}\vert x; y)$ defined as above, it follows that \cref{eq:NA_NS_model_2} holds.

    Finally, to prove that we can always take $p(a; b_0 \ldots b_{n_A - 1}\vert x; y)$ to have uniform marginals on Alice, notice that defining
    \begin{equation}
    \begin{aligned}
        p'(a; b_0 \ldots &b_{n_A - 1} \vert x; y) = \frac{1}{n_A}(p(a; b_0 b_1\ldots b_{n_A - 1} \vert x; y) \\
        & + p(a \oplus 1; b_{n_A - 1} b_0 \ldots b_{n_A - 2} \vert x; y) + \ldots \\
        & \ldots +p(a \oplus (n_A - 1); b_{1} \ldots b_{n_A - 1} b_{0} \vert x; y)),
    \end{aligned}
    \end{equation}
    where $\oplus$ denotes sum modulo $n_A$, we have
    \begin{equation}
    \begin{aligned}
        \sum_{a, b_0\ldots b_{n_{A}-1}} &\delta_{b_a , b}\, p'(a; b_{0}\ldots b_{n_{A}-1}|x;y) =\\
        &\sum_{a, b_0\ldots b_{n_{A}-1}} \delta_{b_a , b}\, p(a; b_{0}\ldots b_{n_{A}-1}|x;y). 
    \end{aligned}
    \end{equation}
\end{proof}

\section{A proof of \cref{global_local_adaptive}}\label{appdx:proof_global_local_adaptive}
\thmadaptivelocalglobal*
\begin{proof}
Assume $\mathbf{q} \in \mathrm{NS}^{\mathrm{A}}(n_X ; n_Y , n_{B})_{n_A}$. Then there exists 
$\mathbf{p} = [p(a,b|x(y,z))]$ a non-signaling box in which Alice has $n_X$ inputs and $n_A$ outputs, while Bob has $n_Y \times n_A$ inputs and $n_B$ outputs, such that
\begin{equation}
q(b|xy) = \sum_a p(ab|x(y,z=a)).
\end{equation}
Fixing $y=i$, the corresponding sub-behavior becomes
\begin{equation}
q_{y}(b|x) = \sum_{a} p(ab|x(y=i,z=a)).
\end{equation}
Define $\mathbf{p}_i = [p(ab|x(i,z))]$. Since $\mathbf{p}$ is NS, each $\mathbf{p}_i$ is also NS and belongs to the Bell scenario $(n_X , n_A ; n_A , n_B)$. Consequently, $\mathbf{p}_i$ provides an adaptive NS-assisted model for $q_{y}(b|x)$ for every $i$.

For the converse direction, we first establish an intermediate lemma, stated below.

\begin{lemma}\label{lemma:uniform_adaptive}
If $\mathbf{q} = [q(b|xy)]$ admits an adaptive NS-assisted model, then there exists an equivalent adaptive NS-assisted model for which Alice's marginals are uniform.
\begin{proof}
Define
\begin{equation}
\tilde{p}(ab|x(yz)) = \frac{1}{n_A} \sum_{i=0}^{n_A-1} 
p(a \oplus_{n_A} i, b \mid x,(y,z \oplus_{n_A} i)),
\end{equation}
where $i \in \{0,\ldots,n_A-1\}$, $\oplus_{n_A}$ denotes addition modulo $n_A$, and $p(ab|x(yz))$ is the original Bell behavior generating $\mathbf{q}$ through
\begin{equation}
q(b|xy) = \sum_a p(ab|x(ya)).
\end{equation}
We show that $\tilde{p}$ indeed defines an alternative adaptive NS-assisted model, as it reproduces $\mathbf{q}$ and remains NS while, additionally, having uniform marginals for Alice.
First, the induced PM behavior is preserved:
\begin{equation}
\begin{aligned}
\sum_a \tilde{p}(ab|x(ya))
&= \sum_a \frac{1}{n_A} \sum_i 
   p(a \oplus_{n_A} i, b|x(y,a \oplus_{n_A} i)) \\
&= \frac{1}{n_A} \sum_i \sum_a 
   p(a \oplus_{n_A} i, b|x(y,a \oplus_{n_A} i)) \\
&= \frac{1}{n_A} \sum_i \sum_a p(ab|x(ya)) \\
&= \frac{1}{n_A} \sum_i q(b|xy)
= q(b|xy),
\end{aligned}
\end{equation}
where we used that, for fixed $i$, the relabeling $a \mapsto a \oplus_{n_A} i$ simply permutes the summation.

Second, Alice's marginal is uniform:
\begin{equation}
\begin{split}
\sum_b \tilde{p}(ab|x(yz))
&= \frac{1}{n_A} \sum_i \sum_b 
   p(a \oplus_{n_A} i, b|x(y,z \oplus_{n_A} i)) \\
&= \frac{1}{n_A} \sum_i p(a \oplus_{n_A} i|x) \\
&= \frac{1}{n_A} \sum_i p(i|x)
= \frac{1}{n_A},
\end{split}
\end{equation}
where we used the NS property of the original behavior and the normalization of $p(i|x)$; the uniform marginal for Alice is itself one of the two NS constraints. 

Finally, for the other NS constraint
\begin{equation}
\begin{split}
\sum_a \tilde{p}(ab|x(yz))
&= \frac{1}{n_A} \sum_i \sum_a 
   p(ab|x(y,z \oplus_{n_A} i)) \\
&= \frac{1}{n_A} \sum_i 
   p(b|y,z \oplus_{n_A} i),
\end{split}
\end{equation}
which depends only on Bob's local variables, as required. Hence, $\tilde{p}$ defines an adaptive NS-assisted model with uniform marginals for Alice.
\end{proof}
\end{lemma}

Let $\mathbf{q} = [q(b|xy)]$ be a PM behavior such that each of its sub-behaviors admits an adaptive NS-assisted model. According to \cref{lemma:uniform_adaptive}, for every $i$ we can write
\begin{equation}
q(b|xy) = q_i(b|x) = \sum_a p_i(ab|x,z=a),
\end{equation}
where $\mathbf{p}_i = [p_i(ab|xz)]$ is NS in the Bell scenario $(n_X , n_A ; n_A , n_B)$. Moreover, we may assume that Alice’s marginals are uniform, i.e.,
\begin{equation}
\sum_b p_i(ab|xz) = \frac{1}{n_A}.
\end{equation}

We now ask whether these Bell sub-behaviors $\mathbf{p}_i$ can be combined into a single behavior on the larger Bell scenario $(n_X , n_A ; n_Y n_A , n_B)$ that remains NS. If so, it would follow that the original PM behavior $\mathbf{q}$ admits an adaptive NS-assisted model. This is indeed possible by defining $p(ab|x(yz)) \equiv p_{i=y}(ab|xz)$. Then
\begin{equation}
\sum_a p(ab|x(yz)) = \sum_a p_{i=y}(ab|xz) = p_{i=y}(b|z),
\end{equation}
which is independent of $x$, and
\begin{equation}
\sum_b p(ab|x(yz)) = \sum_b p_{i=y}(ab|xz) = \frac{1}{n_A},
\end{equation}
which is independent of $y$. Hence, the glued behavior is NS and reproduces $\mathbf{q}$, proving that $\mathbf{q}$ admits an adaptive NS-assisted model.
\end{proof}

\section{The hyperplane representation of $\mathcal{G}(n_X;1,n_B)_{n_A}$}\label{appdx:hyperplane_cci}

The \cref{eq:f_inequalities} holds for any PM scenario $(n_X; n_Y, n_B)_{n_A}$, including those in which Bob has no inputs ($n_Y = 1$). We are particularly interested in this latter case, for which the inequalities reduce to the simpler form
\begin{equation}\label{simplified_CCI}
\sum_b q(b|f(b)) \le n_A .
\end{equation}
For a fixed PM scenario $(n_X; 1, n_B)_{n_A}$, the number of functions $f : [n_B] \to [n_X]$ is $n_X^{n_B}$, which grows exponentially with the scenario size. Nevertheless, a more careful analysis shows that the number of nontrivial inequalities of the form~\eqref{simplified_CCI} is smaller. Here, nontrivial means that the inequality is not already implied by positivity and normalization.

Consider any function $f$ whose image contains at most $n_A$ distinct inputs; equivalently, Bob's output alphabet $[n_B]$ is partitioned into at most $n_A$ subsets, each assigned to a single input of Alice. Any such $f$ generates a trivial inequality. Indeed, consider the extremal case in which $[n_B]$ is partitioned into exactly $n_A$ disjoint subsets $B_0, B_1, \ldots, B_{n_A-1}$ such that $B_i \cap B_j = \emptyset$, $\bigcup_i B_i = B$, $f(b) = x_i$ for all $b \in B_i$, and the inputs $x_i$ are pairwise distinct (not necessarily $x_i = i$). Then,
\begin{equation}
\begin{aligned}
\sum_{b=0}^{n_B-1} q(b|f(b))
&= \sum_{b \in B_0} q(b|x_0) + \cdots + \sum_{b \in B_{n_A-1}} q(b|x_{n_A-1}) \\
&\le \sum_{b=0}^{n_B-1} q(b|x_0) + \cdots + \sum_{b=0}^{n_B-1} q(b|x_{n_A-1}) \\
&= n_A,
\end{aligned}
\end{equation}
where the inequality follows from normalization of each conditional distribution. Therefore, the only functions $f$ that can yield nontrivial inequalities are those whose image contains at least $n_A + 1$ distinct inputs.

The discussion can be further simplified by taking into account the input–output symmetries of the PM scenario. In this case, it suffices to provide a single representative from each symmetry class, since all other inequalities within the same class can be obtained from it by appropriately applying these symmetries. Thus, the task reduces to determining when distinct functions generate inequalities that are inequivalent under such symmetry transformations.

Let $f_1$ be associated with a partition of $B$ into $K \ge n_A + 1$ subsets, namely $B_{0}^{1}, B_{1}^{1}, \ldots, B_{K-1}^{1}$, such that $B_{i}^{1} \cap B_{j}^{1} = \emptyset$, $\bigcup_i B_{i}^{1} = B$, and $f_1(B_{i}^{1}) = x_{i}^{1}$ with $x_{i}^{1} \ne x_{j}^{1}$ for $i \ne j$.\footnote{For the discussion to make sense, $K$ cannot exceed $n_{X}$; otherwise, there would be more subsets than available inputs to assign.} Similarly, let $f_2$ induce a partition of $B$ into $L \ge n_A + 1$ subsets $B_{0}^{2}, B_{1}^{2}, \ldots, B_{L-1}^{2}$ with analogous properties.

If $K \ne L$, the induced inequalities are not equivalent under input–output symmetries. Without loss of generality, assume $K > L$. Since input–output symmetries act only as permutations, they preserve the number of distinct inputs appearing in the inequality. The inequality induced by $f_1$ involves $K$ inputs (namely, $x_{0}^{1}, \ldots, x_{K-1}^{1}$), whereas the one induced by $f_2$ involves only $L < K$ inputs ($x_{0}^{2}, \ldots, x_{L-1}^{2}$). Consequently, no permutation can increase or decrease this number, and hence it is impossible to map one inequality onto the other.

We now assume $K = L$. Even in this case, the induced inequalities need not be equivalent under input–output symmetries. If the parts associated with $f_1$ (namely, $B_{i}^{1}$) have different cardinalities from those associated with $f_2$ (namely, $B_{j}^{2}$), then the corresponding inequalities are inequivalent\footnote{By different cardinalities we mean up to permutation of the parts, so that only the collection of block sizes matters. For instance, $(3,1,1)$ and $(1,3,1)$ are regarded as having the same cardinalities, whereas $(3,1,1)$ and $(2,2,1)$ are not.
}. To see this, note that although both partitions contain the same number of parts—so that $K$ distinct inputs appear in each inequality, namely $x_{0}^{1}, \ldots, x_{K-1}^{1}$ for $f_1$ and $x_{0}^{2}, \ldots, x_{K-1}^{2}$ for $f_2$—the multiplicities with which each input appears may differ. Specifically, $x_{i}^{1}$ appears $|B_{i}^{1}|$ times, whereas $x_{i}^{2}$ appears $|B_{i}^{2}|$ times. Since input symmetries act only by permutations, they preserve these multiplicities and therefore cannot transform one inequality into the other.

Now suppose $K=L$ and that, for every part $B_i^{1}$ of the partition associated with $f_1$, there exists a part $B_j^{2}$ of the partition associated with $f_2$ such that $|B_i^{1}| = |B_j^{2}|$. In this case, the resulting inequalities are different (as algebraic expressions, since $f_1 \neq f_2$), but they are equivalent under input–output symmetries. Indeed, each output $b\in B$ appears exactly once in both inequalities, and both involve the same number $K \ge n_A+1$ of distinct inputs. Moreover, the multiplicity with which each input appears is the same in both cases. Consequently, suitable permutations of inputs and outputs map one inequality onto the other, and vice versa, rendering them indistinguishable under the symmetry transformations of the scenario.

Therefore, all functions $f$ associated with a partition of $B$ into 
$K \ge n_A+1$ subsets $(B_0^f,\ldots,B_{K-1}^f)$ whose cardinalities, once arranged
in nonincreasing order,
$|B_0^f| \ge |B_1^f| \ge \cdots \ge |B_{K-1}^f|$,
coincide, generate inequalities belonging to the same symmetry class. 
Hence, each symmetry class is completely characterized by the number of parts $K$
together with the corresponding subset cardinalities. Consequently, for a PM
scenario $(n_X;1,n_B)_{n_A}$, the nontrivial inequalities invariant under
symmetries are obtained simply by enumerating the integer partitions of $n_B$
with $K \ge n_A+1$ and $K \le n_X$, since at most $n_X$ distinct inputs of Alice
can appear.

For instance, consider the PM scenario $(4;1,4)_{2}$. There are $4^{4}=256$
functions $f$ (and hence the same number of inequalities). However, once the
input–output symmetries of the PM scenario are taken into account, only two
inequivalent symmetry classes remain. Indeed, $n_B=4$ admits the five integer partitions
$(4)$, $(3,1)$, $(2,2)$, $(2,1,1)$, and $(1,1,1,1)$. Among these, only the last
two have at least $n_A+1=3$ parts (and at most $n_X=4$ parts), and therefore
lead to nontrivial inequalities. From each partition a representative inequality can be constructed. The
partition $(2,1,1)$ means that two outputs are mapped to the same input $x_0$,
while the remaining outputs are mapped to distinct inputs $x_1$ and $x_2$.
One possible representative is
\begin{equation}
q(0|0)+q(1|1)+q(2|2)+q(3|0)\le 2.
\end{equation}
For the partition $(1,1,1,1)$, each output is mapped to a distinct input, yielding,
for instance,
\begin{equation}\label{test_ineq}
q(0|0)+q(1|1)+q(2|2)+q(3|3)\le 2.
\end{equation}

\end{document}